\pdfoutput=1 

\documentclass{article}
\usepackage{ijcai19}

\usepackage{amssymb}
\usepackage{mathtools}
\usepackage{framed}
\usepackage{cases}
\usepackage[noend]{algpseudocode}  
\usepackage{multirow}
\usepackage{subfigure}

\usepackage{times}
\usepackage{soul}
\usepackage{url}
\usepackage{hyperref}
\usepackage[small]{caption}
\usepackage{graphicx}
\usepackage{amsmath}
\usepackage{amsthm}
\usepackage{booktabs}
\usepackage{algorithm}
\usepackage{color}
\title{Collective Mobile Sequential Recommendation: A Recommender System for Multiple Taxicabs}

\author{
	Tongwen Wu,
	Zizhen Zhang\footnote{Corresponding Author}, Yanzhi Li \And Jiahai Wang\\
	\affiliations
	School of Data and Computer Science, Sun Yat-sen University, China \emails
	wutw@mail2.sysu.edu.cn,
	zhangzzh7@mail.sysu.edu.cn
}

\begin{document}
	\maketitle
	\newtheorem{definition}{Definition} 
	\newtheorem{theorem}{Theorem}

\begin{abstract}
Mobile sequential recommendation was originally designed to find a promising route for a single taxicab. Directly applying it for multiple taxicabs may cause an excessive overlap of recommended routes. The multi-taxicab recommendation problem is challenging and has  been less studied. In this paper, we first formalize a collective mobile sequential recommendation problem based on a classic mathematical model, which characterizes time-varying influence among competing taxicabs. Next, we propose a new evaluation metric for a collection of taxicab routes aimed to minimize the sum of potential travel time. We then develop an efficient algorithm to calculate the metric and design a greedy recommendation method to approximate the solution. Finally, numerical experiments show the superiority of our methods. In trace-driven simulation, the set of routes recommended by our method significantly outperforms those obtained by conventional methods. 
\end{abstract}

\section{Introduction}
In smart cities, the prevalence of wireless sensors and communication infrastructure such as GPS, Wi-Fi and RFID makes large-scale trace data available. This enables us to mine useful knowledge of taxicab transporting systems and passenger patterns. The extracted knowledge in turn assists in designing intelligent strategies to increase taxicab drivers' profit and shorten passengers' waiting time.

In this paper, we study a mobile recommender system, which can provide a set of promising routes for a collection of taxicab drivers. The ultimate target is to systematically reduce the vacant taxicab's cruising time from a global perspective and in a real-time manner. There are some essential problems within the scope of mobile recommender system \cite{Zheng2010}. One interesting and practical problem is \textbf{mobile sequential recommendation (MSR)} \cite{Ge2010An}. In MSR, the locations where pick-up events occur are clustered into pick-up points. A route, i.e., a sequence of pick-up points, is to be recommended to a taxicab driver such that his expected cruising distance before having passengers is minimized.

	While the route recommendation for a single taxicab has been thoroughly considered  \cite{where-to-find-my-next-passenger,Yuan2013T,Zhang2012pCruise,Huang2014Backward,Yun2011An}, how to incorporate the influence among routes of multiple taxicabs and maximize the overall profit still remains challenging. 
	To this end, some methods have been proposed. 
	In \cite{Ge2010An}, by maintaining top $K$ routes in a buffer, the recommender system randomly chooses a route for each taxicab in the same area.
	\cite{Qu2014A} provided an improved top-K method which considers some correlation among routes.
	\cite{Qian2015SCRAM} introduced a route assignment mechanism aimed to achieve recommendation fairness for a group of taxicab drivers. 
	\cite{Ye2018} firstly formalized a multi-user mobile sequential recommendation problem in which the recommended routes are required to be disjoint.
	
	However, to the best of our knowledge, there is no existing study considering the time-varying competition among taxicab drivers and incorporating it into the multi-taxicab recommendation. For this challenge, we construct a more practical and complex MSR problem called \textbf{collective mobile sequential recommendation (CMSR)}, which targets at generating multiple interrelated routes for a group of taxicabs.

	Our work and contributions can be summarized as follows. Firstly, we use a classic mathematical model to characterize the time-varying pick-up probability, which considers the taxicab competitions and passenger arrival patterns. To be more specific, by modelling the passenger arrival pattern as a Poisson process, we generalize the pick-up probability of a pick-up point, instead of being a constant, into a function of the time interval between two consecutive taxicab arrivals. 
	Secondly, in order to maximize the profit globally, we introduce the sum of potential travel time for a collection of taxicabs as a metric to evaluate a multi-taxicab recommendation. As CMSR is a harder combinatorial problem than traditional MSR, it is computational more intensive to find an optimal/sub-optimal solution. After showing that evaluating a given recommendation in a straightforward manner requires exponential time, we propose an alternative Sequential Evaluation approach which requires a lower time complexity.
	Finally, we design a greedy approach to obtain an approximate solution, which performs very well in extensive numerical experiments and trace-driven simulation.
		
\section{Problem Formulation}
	
	
\subsection{Preliminary: MSR Formulation}
	
In literature, a classic MSR problem is described as follows. Let $\mathcal{C} = \{1,2,...,N\}$ be a set of central pick-up points and 0 be the initial position of a taxicab. Denote by $D(c,c^{'})$ the distance between two pick-up points $c$ and $c^{'}$, and $P(c)$ the estimated probability that a pick-up event occurs at a pick-up point $c$. A route $\overrightarrow{R} = (c_{1},c_{2},...,c_{L})$ is a directed sequence generated from a subset of $\mathcal{C}$ of length $L$. Note that the points in $\overrightarrow{R}$ are generally different from each other. All possible route sequences of size $L$ constitute a feasible solution space, denoted by $\mathbb{R}^{L}$. The expected cruising distance of a taxicab before picking up customers is recognized as the potential travel distance (PTD), which can be computed as follows:

\begin{equation}
	\small
	\mathrm{PTD}(\overrightarrow{R}) = \overrightarrow{d} \cdot \overrightarrow{p},
\label{equ:ptd}
\end{equation}
 where $\overrightarrow{d}$ is a distance vector and $\overrightarrow{p}$ is a probability vector. These vectors are given by:
\begin{equation}
	\small
	\begin{split}
		\overrightarrow{d} = (D(0,c_{1}),D(0,c_{1})+D(c_{1},c_{2}),...,
		D(0,c_{1})+\\
		\sum_{j=1}^{L-1}D(c_{j},c_{j+1}),D(0,c_{1})+\sum_{j=1}^{L-1}D(c_{j},c_{j+1})+Penalty)
	\end{split}
\end{equation}
\begin{equation}
	\small
	 	\overrightarrow{p} = (P(c_{1}),\overline{P(c_{1})}P(c_{2}),...,\prod_{j=1}^{L-1}\overline{P(c_{j})}P(c_{L}),\prod_{j=1}^{L}\overline{P(c_{j})})	
\end{equation}

In Equation (2), $Penalty$ represents a certain penalty distance of not picking up any passengers along the route. In Equation (3), $\overline{P(c)}=1-P(c)$. The MSR problem is to recommend an optimal driving route $\overrightarrow{R} (\overrightarrow{R} \in \mathbb{R}^{L}) $ such that the corresponding PTD is minimized.


\subsection{CMSR Formulation}

The classic MSR is aimed to recommend a route to each taxicab driver independently. However, in the real world, the recommendation should respond to requests from multiple users, namely, a CMSR problem. Directly applying MSR for each user will lead to  excessive overlaps among the recommended routes \cite{Ye2018}. To tackle CMSR, we need to consider the influence among taxicabs and recommend routes from a global perspective.

Consider a scenario that passengers arrive at a pick-up point randomly and wait for taxicabs to come. In the concerned time period, the passenger arrival pattern is modeled as a homogeneous Poisson process. Specifically, we can estimate a \textbf{passenger arrival rate} $\lambda_{c}$ for point $c$. Let $N(t)$ be the number of passenger arrival events in the time interval $[0,t]$ at point $c$. Then, the number of arrivals in time interval $[t,t+\Delta]$ follows a Poisson distribution as follows:
\begin{equation}
\small
	P\{N(t+\Delta)-N(t)=i\} = \frac{\mathrm{e}^{-\lambda_{c}\cdot\Delta}(\lambda_{c}\cdot\Delta)^{i}}{i!} (i=0,1,2,...)     
	\label{equ:poi}
\end{equation}

We assume that passengers waiting at a pick-up point for the upcoming taxicab will not be picked up by the next arrival taxicab. In other words, if a passenger $p$ arrives at time $t_{p}$ and two taxicabs $i$ and $j$ arrive at $t_{i}$ and $t_{j}$ ($t_{p}<t_{i}<t_{j}$), then taxicab $j$ cannot pick up passenger $p$. Therefore, we can reformulate the probability that a taxicab arriving at $c$ picks up passengers after an interval $\Delta$ since the last taxicab comes. It can be computed as:
\begin{equation}
	\small
	\begin{split}
	P(c,\Delta) &= \sum_{i=1}^{\infty} \frac{\mathrm{e}^{-\lambda_{c}\cdot\Delta}(\lambda_{c}\cdot\Delta)^{i}}{i!} \\
					&= 1-\mathrm{e}^{-\lambda_{c}\cdot\Delta}
	\end{split}
\end{equation}

Note that if a taxicab is the first one to visit the pick-up point at time $t$, we have $\Delta = t$. Specially, $P(0,\Delta)=0$ for any $\Delta$.

As the probability is related to time in CMSR, the distance measurement $D(c,c^{'})$ is replaced by $T(c,c^{'})$ to indicate the traveling time between pick-up points $c$ and $c^{'}$. In practice, $T(c,c^{'})$ can be discretized into an integer, e.g., a second, for ease of computation. Now, consider recommending a collection of $K$ routes of length $L$ to taxicabs at the same area (point 0) and at the same time (time 0). A possible recommendation $\mathcal{R}$ is an \emph{ordered multiset} of routes from $\{\overrightarrow{R_{1}},\overrightarrow{R_{2}},...,\overrightarrow{R_{K}}\}$ chosen from $\mathbb{R}^{L}$. The ordered set is helpful in the following scenario: if two taxicabs $\overrightarrow{R_{i}}$ and $\overrightarrow{R_{j}}$ $(i<j)$ arrive at a pick-up point at the same time, we assume that the taxicab of $\overrightarrow{R_{i}}$ arrives earlier.

We use $\mathcal{F}(\mathcal{R})$ to denote the sum of \textbf{potential travel time (PTT)} of $K$ routes. The CMSR can be formalized as follows:
	\begin{framed}
		\noindent \textbf{The CMSR problem} \\
		\textbf{Given:} A set of pick-up points $\mathcal{C}$ of size $N$, a set $\mathbb{R}^{L}$ of all possible routes of length $L$, $\lambda_{c}$ for each pick-up point, the inital position 0 for $K$ taxicabs. \\ 
		\textbf{Objective:} Recommending an optimal set of driving routes  $\mathcal{R} (\overrightarrow{R_{k}} \in \mathbb{R}^{L})$ with the goal to minimize the sum of PTT:
		\begin{equation*}
		\min \mathcal{F}(\mathcal{R})
		\end{equation*}
	\end{framed}
	\begin{figure}
		\centering
		\includegraphics[width=0.75\linewidth]{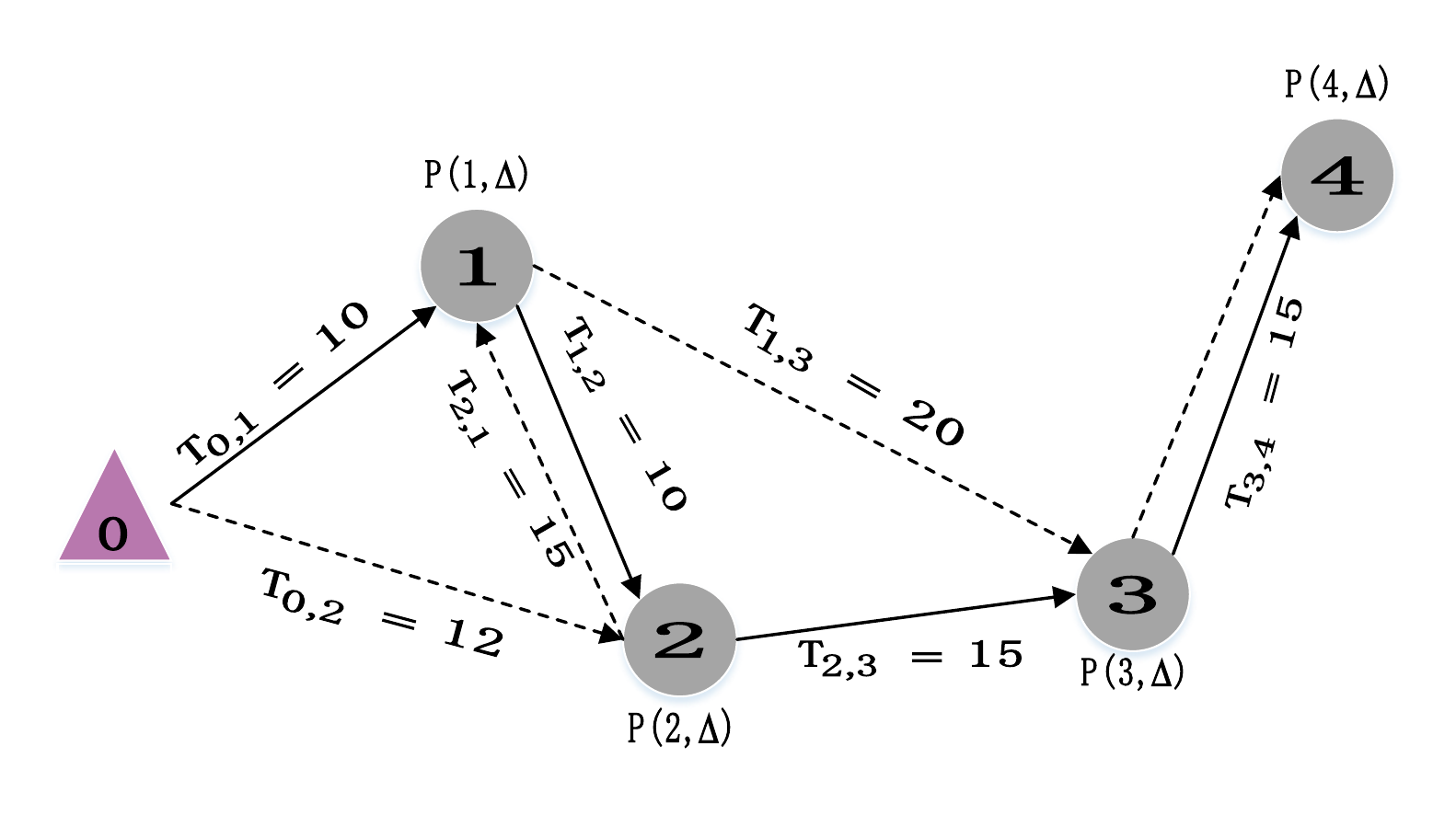}
		\caption{An illustrated example of CMSR.}
		\label{fig:ill}
	\end{figure}

The computation of $\mathcal{F}(\mathcal{R})$ is non-trivial. It is much more complex than computing a single taxicab's PTD in two aspects. First, the pick-up probability varies with the arrival time of taxicabs. Second, the distribution of cruising time of a route is correlated with the others.

	\paragraph{Example.}
	An illustration example is shown in \figurename~\ref{fig:ill}, in which two routes are recommended. The first route is represented by a solid line and the second one is represented by a dash line. As  shown, the two taxicabs will arrive at their corresponding first pick-up points $1$ and $2$. If either of them fails to pick-up a passenger, it will go ahead to its next pick-up point. In this case, the pick-up probability that the first taxicab ends up cruising at $2$ is $\overline{P(1,10)}P(2,8)$, the probability that second taxicab ends up cruising at $1$ is $\overline{P(2,12)}P(1,17)$. Next, consider the event that the second taxicab ends up cruising at $3$. If the first taxicab has picked up a passenger at $1$ or $2$, the conditional probability of the event is $\overline{P(2,12)}\ \overline{P(1,17)}P(3,47)$. Otherwise, if the first taxicab arrives at 3 before the second taxicab, the conditional probability is $\overline{P(2,12)}\ \overline{P(1,17)}P(3,12)$. Thus, the probability that second taxicab ends up cruising at $3$ is $(P(1,10)+\overline{P(1,10)}P(2,8))\overline{P(2,12)}\ \overline{P(1,17)} P(3,47) +(1-P(1,10)-\overline{P(1,10)}P(2,8)) \overline{P(2,12)}\ \overline{P(1,17)}P(3,12)$.
	
\section{Methodology}
In CMSR, the number of possible recommendations is $(\binom{N}{L}L!)^{K}$, which is exponentially increasing with $K$ and based on the number of all possible routes $\binom{N}{L}L!$ in MSR. This combinatorial problem requires extremely high computational resources.	To address this challenge, we propose an algorithm based on dynamic programming to efficiently calculate $\mathcal{F}$ for a given $\mathcal{R}$. Next, we propose a greedy recommendation algorithm to approximate the solution from a global perspective.

\subsection{Evaluating a Collective Recommendation}
We first give a straightforward approach for computing $\mathcal{F}$. By characterizing some useful properties among different recommendations, we can design a more efficient approach.
	
Consider a recommendation $\mathcal{R}$. Let $\textbf{u}=(u_{1},u_{2},...,u_{K})$ denote its possible outcome, where $u_{k}$ means that the taxicab of $\overrightarrow{R_{k}}$ ends up cruising at the $u_{k}^{th}$ pick-up point. In particular, $u_{k}=L+1$ indicates that the taxicab fails to pick-up passengers along the whole route. Thus, $u_{k}$ is larger than zero and less than or equal to the length of $\overrightarrow{R_{k}}$ plus one. We use $\mathbb{U}$ to denote the set of all possible outcomes. Let $p(\textbf{u})$ and $s(\textbf{u})$ be the probability that $\textbf{u}$ occurs and the total cruising time of $\textbf{u}$.

\paragraph{Example}
Take \figurename~\ref{fig:ill} as an example. $\mathbb{U}$ is $\{(i,j)| 1\le i,j \le 5 \}$. $\textbf{u}=(1,2)$ means that both  taxicabs pick up passengers at point 1. Then, $\textbf{p}((1,1)) = P(1,10)P(2,12)$, $\textbf{p}((2,1)) = \overline{P(1,10)}P(2,8)P(2,12)$, $\textbf{p}((2,2))=\overline{P(1,10)}P(2,8)\overline{P(2,12)}P(1,17)$, $\textbf{p}((2,3)) = \overline{P(1,10)}P(2,8)\overline{P(2,12)}\ \overline{P(1,17)}P(3,47)$ and $\textbf{s}((2,3)) = T(0,1)+T(1,2) +T(0,2)+T(2,1)+T(1,3)$.

According to the definition, $\mathcal{F}$ can be calculated by summing up the PTT of each route. However, because the cruising time of one route is dependent on the outcome of the other routes, we instead compute $\mathcal{F}$ by summing up the product of outcome event probability and the total cruising time of $\textbf{u}$ as below:
	\begin{equation}
		\small
		\mathcal{F}(\mathcal{R}) = \sum_{\textbf{u} \in \mathbb{U}}p(\textbf{u})*s(\textbf{u}) 
	\end{equation}
	
For ease of composition, we use $c_{k,u}$ and $t_{k,u}$ to denote the $u^{th}$ pick-up point of $\overrightarrow{R_{k}}$ and its corresponding arrival time, respectively. Based on the above discussions, we can form several tuples $(k,u,c_{k,u},t_{k,u})$ for a given recommendation $\mathcal{R}$. Note that if $u$ is equal to $L+1$, $c_{k,u}$ does not exist, and therefore such kind of tuples are not considered.

Given a certain $\textbf{u}$, we next consider how to calculate $p(\textbf{u})$ and $s(\textbf{u})$. $p(\textbf{u})$ can be computed by the joint probability that every event occurs. Here, an event is that a taxicab picks up passengers at a pick-up point or not. Because the occurred  events are known when $\textbf{u}$ is determined, we can obtain the time intervals between the arrival of adjacent taxicabs. With these time intervals, we can get the probability of each event. To be more specific, let $\Delta_{k,u}$ denote the time interval between the arrival of taxicab $k$ and the last taxicab visiting at $c_{k,u}$. For example, in Figure \ref{fig:ill}, if $\textbf{u}=(2,3)$, then $\Delta_{1,2} = 20-12$, $\Delta_{2,2} = 27-10$ and $\Delta_{2,3} = 47$.

We turn to seek the value of $\Delta_{k,u}$ for each tuple. Consider a particular pick-up point $c$. Find all the tuples $(k,u,c_{k,u},t_{k,u})$ with $c=c_{k,u}$ and $u \leq u_{k}$ (i.e., taxicab $k$ is still cruising when it reaches the $u$-th pick-up point $c$). Sort them by the value of $t_{k,u}$ primarily and the value of $k$ secondarily in ascending order. Then we can get an ordered sequence of tuples at point $c$. If the $p^{th}$ tuple is $(k,u,c,t)$ and the $(p-1)^{th}$ tuple is $(k',u',c,t')$, then $\Delta_{k,u}$ is equal to $t-t'$. Finally, $p(\textbf{u})$ and $s(\textbf{u})$ can be computed as follows:
\begin{equation}
	\small
		\begin{split}
		p(\textbf{u}) = \prod_{k=1}^{K} \big{(}\prod_{u=1}^{u_{k}-1}\overline{P(c_{k,u},\Delta_{k,u})} (P(c_{k,u_{k}},\Delta_{k,u_{k}})\\
		\mathbb{I}_{u_{k} \leq L}+\mathbb{I}_{u_{k} = L+1})\big{)}
		\end{split}
\end{equation}
\begin{equation}
	\small
		\begin{split}
		s(\textbf{u}) = \sum_{k=1}^{K} \big{(}\sum_{u=2}^{u_{k}-1} T(c_{k,u-1},c_{k,u})+T(0,c_{k,1}) + \\ T(c_{k,u_{k}-1},c_{k,u_{k}})\mathbb{I}_{u_{k} \leq L} + Penalty\mathbb{I}_{u_{k} = L+1}\big{)},
		\end{split}
\end{equation}	
 where $\mathbb{I}_{condition}$ is equal to one if the condition holds, and zero otherwise.

\subsubsection{A Straightforward Approach}
It is straightforward that we can compute $p(\textbf{u})$ and $s(\textbf{u})$ for each $\textbf{u}$ independently. The pseudo code is shown in Algorithm \ref{alg:2}. In this algorithm, we enumerate all possible outcomes and compute the probability and the total cruising time for each $\textbf{u}$. Lines 7-13 show the procedure of computing the probability $p(\textbf{u})$. The time complexity of this procedure is $\mathcal{O}(KL)$. Similarly, the cruising time can also be computed in $\mathcal{O}(KL)$ time complexity. The number of all possible outcomes is $\mathcal{O}((L+1)^K)$. Thus, the time complexity of this approach is $\mathcal{O}(KL(L+1)^K)$. 


	\begin{algorithm}[htb] 
		\small
		\caption{A Straightforward Approach} 
		\label{alg:2}  
		\begin{algorithmic}[1]  
			\Require A set of pick-up points $\mathcal{C}$, a probability set $\mathcal{P}$, initial position, the travel time matrix $T$ and a recommendation $\mathcal{R}$ in which $\textbf{u}=(u_{1},u_{2},...,u_{K})$
			\Ensure The sum of potential travel time
			\State Make tuples $(k,u,c_{k,u},t_{k,u})$ and sort them by $t$ (primary key) and $k$ (secondary key) in increasing order
			\State Create an array $\mathrm{last}$, $\mathrm{last}_{c}$ denotes the last visited time of $c$, the initial value is $0$
			\State $\mathrm{ans} = 0$
			\For{each $\textbf{u}$}
			\State Set $\mathrm{last}_{c}$ as 0 for every $c$
			\State $p = 1$
			\For{each tuple $(k,u,c,t) $ in order}
				\If{$u < u_{k}$}
				\State	$p=p*\overline{P(c,t-\mathrm{last}_{c})}$
				\EndIf
				\If{$u=u_{k}$}
				\State	$p=p*P(c,t-\mathrm{last}_{c})$
				\EndIf
				\If{$u\leq u_{k}$}
				\State	$\mathrm{last}_{c} = t$
				\EndIf
			\EndFor
			\State calculate the sum of cruising time $s$
			\State $\mathrm{ans} = \mathrm{ans} + p*s$
			\EndFor
			\State \Return ans
		\end{algorithmic}  
	\end{algorithm}  
	
\subsubsection{An Improved Approach}
The straightforward approach is not efficient enough. Observing that the computation of different $\textbf{u}$ has a lot of common parts, the overlapping subproblems can be calculated just once and reused multiple times by recursive equations. Therefore, we can design an algorithm based on dynamic programming to accelerate the computation of $p(\textbf{u})$ and $s(\textbf{u})$. 
	
Consider a more general scenario where the route lengths of a recommendation are not necessarily the same (but still no greater than $L$). In the following, we show how different recommendations correlate to the same outcome $\textbf{u}$. 
To do so, it is necessary to extend the previous notations. Let $p_{\mathcal{R}}(\textbf{u})$ and $s_{\mathcal{R}}(\textbf{u})$ respectively denote the probability and the total cruising time for recommendation $\mathcal{R}$.

\begin{theorem}[SEQUENTIAL EQUATION FOR $p$]
		Let $\mathcal{R}$ and $\mathcal{R^{'}}$ be two recommendations which only differ in the $q^{th}$ route. Suppose that $\mathcal{R}=\{\overrightarrow{R_{1}},\overrightarrow{R_{2}},...,\overrightarrow{R_{K}}\}$ and the length of $\overrightarrow{R_{k}}$ is $l_{k}$. The $q^{th}$ route of $\mathcal{R}$ is $(c_{1},c_{2},...,c_{l-1})$ and the $q^{th}$ route of $\mathcal{R^{'}}$ is $(c_{1},c_{2},...,c_{l})$. And let $t_{\mathrm{end}}$ be the arrival time when the taxicab of $\overrightarrow{R'_{q}}$ arrives at $c_{l}$.
		  
		Assume that $t_{end}$ satisfies the following condition:
		\begin{equation}
		\small
		\begin{split}
			t_{k,u}<t_{\mathrm{end}} \ or\  (t_{k,u}=t_{\mathrm{end}} \ and\  k<q), \\ 
			\text{$\forall~k \ne q$ and $c_{k,u}=c_{l}$} 
		\end{split}
		\label{equ:con}
		\end{equation}

		In other words, the latest visiting time of $c_{l}$ is $t_{\mathrm{end}}$ by the $q^{th}$ route. Other routes either visit $c_l$ before $t_{end}$ or not visit $c_l$.	Then $p_{\mathcal{R^{'}}}(\textbf{u})$ satisfies the following equation:
		\begin{equation}
		\small
			p_{\mathcal{R^{'}}}(\textbf{u}) =
			\begin{cases}
			p_{\mathcal{R}}(\textbf{u}), & 0 < u_{q} < l \\
			p_{\mathcal{R}}(\textbf{u})P(c_{l},\Delta_{q,l}),  & u_{q}=l  \\
			p_{\mathcal{R}}(\textbf{u}-\textbf{i}_{q})\overline{P(c_{l},\Delta_{q,l})} & u_{q} = l+1,  \\
			\end{cases}
			\label{equ:seq}
		\end{equation}
		 where  $\textbf{i}_{q}$ denotes the vector where all the elements are 0 except that the $q^{th}$ element is 1.
		\label{theo:seq}
\end{theorem}
\begin{proof}
		
		Note that tuples $(k,u,c_{k,u},t_{k,u})$ in $\mathcal{R}$ are also included in $\mathcal{R^{'}}$, but $\mathcal{R^{'}}$ contains one more tuple $(q,l,c_{l},t_{\mathrm{end}})$.
		
		For the first case in Equation (\ref{equ:seq}), the outcome $\textbf{u}$ reflects the same cruising routes for both $\mathcal{R}$ and $\mathcal{R^{'}}$. Thus their probabilities are equal.

		For the second case, the outcome of $\mathcal{R^{'}}$ is slightly different from the outcome of $\mathcal{R}$. The taxicab of $\overrightarrow{R'_{q}}$ continues its driving to $c_{l}$ after it fails to pick up passengers at the last pick-up point $c_{l-1}$ of $\overrightarrow{R_{q}}$, but the taxicab of $\overrightarrow{R_{q}}$ has already ended its route and a penalty has been added to the total cruising time. Observe that whenever the taxicab of $\overrightarrow{R'_{q}}$ arrives at $c_{l}$, it is the last one to visit according to (\ref{equ:con}). Therefore, $\mathcal{R}$ and $\mathcal{R^{'}}$ share the same time intervals. Only $\Delta_{q,l}$ needs to be figured out and appended to $\mathcal{R^{'}}$. 
		Then we have the following equation:
		\begin{equation}
		\small
			\begin{split}
			p_{\mathcal{R^{'}}}(\textbf{u})
			&= \prod_{k=1,k \ne q}^{K} \big{(} \prod_{u=1}^{u_{k}-1}\overline{P(c_{k,u},\Delta_{k,u})} (P(c_{k,u_{k}},\Delta_{k,u_{k}})\mathbb{I}_{u_{k} \leq l_{k}} \\ 
			&\quad +\mathbb{I}_{u_{k} = l_{k}+1})\big{)} \prod_{u=1}^{l-1}\overline{P(c_{q,u},\Delta_{q,u})}P(c_{l},\Delta_{q,l})\\
			&= p_{\mathcal{R}}(\textbf{u})  P(c_{l},\Delta_{q,l})\\
			\end{split}
		\end{equation}
		The proof of the third case is similar to the second case.
\end{proof}

\paragraph{Example} We take Figure \ref{fig:ill} as an example. The recommendation shown in the figure is $\{(1,2,3,4), (2,1,3,4)\}$. Consider another two recommendations with their lengths not necessarily equal to $L$: $\mathcal{R}=\{(1,2,3),(2,1)\}$ and $\mathcal{R^{'}}=\{(1,2,3),(2,1,3)\}$.
\begin{itemize}
\item If $\textbf{u}=(2,2)$, then $p_{\mathcal{R^{'}}}(\textbf{u})=p_{\mathcal{R}}(\textbf{u})=\overline{P(1,10)}P(2,8)\overline{P(2,12)}P(1,17)$.
\item If $\textbf{u}=(2,3)$, then $p_{\mathcal{R^{'}}}(\textbf{u})=p_{\mathcal{R}}(\textbf{u})P(3,47)=\overline{P(1,10)}P(2,8)\overline{P(2,12)}\ \overline{P(1,17)}P(3,47)$.
\item If $\textbf{u}=(2,4)$, then $p_{\mathcal{R^{'}}}(\textbf{u})=p_{\mathcal{R}}((2,3))\overline{P(3,47)}$.
\end{itemize}
	
	\begin{theorem}[SEQUENTIAL EQUATION FOR $s$]
		Under the same condition in Theorem \ref{theo:seq}, $s$ satisfies the following equation:
		\begin{equation}
			\small
			s_{\mathcal{R^{'}}}(\textbf{u}) =
			\begin{cases}
			s_{\mathcal{R}}(\textbf{u}), & 0 < u_{q} < l \\
			s_{\mathcal{R}}(\textbf{u})-Penalty+T(c_{l-1},c_{l}),  & u_{q}=l  \\
			s_{\mathcal{R}}(\textbf{u}-\textbf{i}_{q})+T(c_{l-1},c_{l}) & u_{q} = l+1\\
			\end{cases}
		\end{equation}
		\label{theo:seq2}
	\end{theorem}
The proof of Theorem \ref{theo:seq2} is similar to Theorem \ref{theo:seq}. The calculation of cruising time is even simpler, so we omit the proof here.

Theorem \ref{theo:seq} and \ref{theo:seq2} characterize the sequential relation of outcomes between two recommendations. Based on sequential equations, we can calculate $p$ and $s$ sequentially for a series of auxiliary recommendations, and finally achieve the evaluation of recommendation $\mathcal{R}$. Specifically, sort all tuples $(k,u,c_{k,u},t_{k,u})$ for $\mathcal{R}$ by first $t_{k,u}$ and then $k$ in increasing order. A prefix of sorted tuple sequence forms an auxiliary recommendation. It is obvious that the consecutive recommendations satisfy the condition in Theorem \ref{theo:seq}. The pseudo code for the improved approach, called a sequential evaluation approach, is presented in Algorithm \ref{alg:3}.

	\begin{algorithm}[htb]  
		\small
		\caption{A Sequential Evaluation Approach} 
		\label{alg:3}  
		\begin{algorithmic}[1]  
			\State Create two $K$-dimensional arrays $p$ and $s$, one-dimensional array $l$
			\State Set initial values of $p,s,l$ as 0
			\State Set initial outcome $\textbf{u}=(1,1,...,1)$
			\State $p(\textbf{u}) = 1$; $s(\textbf{u}) = K \cdot Penalty$
			\State Make tuples $(k,u,c_{k,u},t_{k,u})$ and sort them by $t$ and then $k$ in increasing order
			\For{each tuple $(k,u,c,t)$ in order}
			\State Create variable $last$ denoting the last visit time of $c$ before taxicab of $\overrightarrow{R_{k}}$ arrives at $c$, initial value is 0
			\For{each $\textbf{u}$ satisfies $u_{k}=l_{k}+1$ and $u_{j} \leq l_{j} +1 (j \ne k)$ }
			\State Update $last$
			\State $p(\textbf{u}+\textbf{i}_{k}) = p(\textbf{u}) *(1- P(c,t-last))$
			\State $p(\textbf{u}) = p(\textbf{u}) * P(c,t-last)$
			\State $s(\textbf{u}+\textbf{i}_{k}) = s(\textbf{u}) + T(c_{k,u-1},c)$
			\State $s(\textbf{u}) = s(\textbf{u}) -Penalty + T(c_{k,u-1},c)$
			\EndFor
			$l_{k} = l_{k} + 1$ 
			\EndFor
			\State $ans = 0$
			\For{\textbf{each} $\textbf{u} \in \mathbb{U}$}
			\State $ans = ans + p(\textbf{u})*s(\textbf{u})$
			\EndFor
			\State \Return $ans$
		\end{algorithmic}  
	\end{algorithm}  
	\paragraph{Complexity Analysis} 
	In the inner loop (lines 8--13), every time when $\textbf{u}$ gets accessed, a new outcome never accessed $\textbf{u}+\textbf{i}_{k}$ becomes accessed. The number of all outcomes is $\mathcal{O}((L+1)^K)$ and the maintenance of variable $last$ is easily implemented in $\mathcal{O}(1)$ time. So the whole complexity of outer and inner loop is $\mathcal{O}((L+1)^{K})$. The loop of $\textbf{u}$ in lines 15--16 is $\mathcal{O}((L+1)^K)$. Thus, the time complexity of algorithm \ref{alg:3} is $\mathcal{O}((L+1)^{K})$. It achieves a lower complexity than Algorithm \ref{alg:2} with a significant margin.

\subsection{Seeking a Collective Recommendation}
In \S 3.1, we see that the complexity of evaluating a possible recommendation grows exponentially with the number of taxicabs $K$. Finding an optimal recommendation of CMSR is even challenging. Note that powerful heuristic approaches such as tabu search and genetic algorithm for combinatorial problems are not suitable for CMSR, as search operators would perform a lot of times and involve many recommendation evaluations. In order to solve the problem within a reasonable time, we propose a greedy algorithm to approximate the optimal solution. The basic idea is as follows. If we want to append a pick-up point to the end of an incomplete route with its length smaller than $L$, we can find the best choice by enumerating and evaluating all the combinations of the pick-up points and routes. In this way, by adding pick-up points iteratively, we can finally obtain a feasible recommendation. 

The proposed greedy recommendation approach is shown in Algorithm \ref{alg:4}, which is a concise framework. In line 8, it generates one candidate recommendation by appending one pick-up point to one route. Note that the candidate recommendation and its previous recommendation do not need to satisfy the condition in Theorem \ref{theo:seq} and \ref{theo:seq2}. In line 9, it evaluates a new recommendation by simply calling Algorithm \ref{alg:3}. The algorithm terminates when all the routes are of length $L$.

	\begin{algorithm}[htb]  
	\small
		\caption{Greedy Recommendation for CMSR} 
		\label{alg:4}  
		\begin{algorithmic}[1]  
			\Require A set of pick-up points $\mathcal{C}$, a probability set $\mathcal{P}$, initial position, the length $L$, the travel time matrix $T$ and the number of taxicabs $K$
			\Ensure A recommendation $\mathcal{R}$
			\State $\mathcal{R} = (\overrightarrow{R_{1}},\overrightarrow{R_{2}},...,\overrightarrow{R_{K}})$, where $\overrightarrow{R_{k}} = \emptyset$
			\Repeat
			\State $min = \infty$; 
			\For{$j=1$ to $K$}
			\If{length of $\overrightarrow{R_{j}}$ is less than $L$}
			\For{$i=1$ to $N$}
			\If{point $i$ is not in  $\overrightarrow{R}_{j}$}
			\State Let $\mathcal{R}^{'}$ be the candidate recommendation that $i$ is appended to $\overrightarrow{R}_{j}$
			\State Evaluate $\mathcal{R^{'}}$
			\If{$\mathcal{F}(\mathcal{R^{'}}) < min$}
			\State $min = \mathcal{F}(\mathcal{R^{'}})$
			\State $\mathcal{R}_{\text{min}} = \mathcal{R}^{'}$
			\EndIf
			\EndIf
			\EndFor
			\EndIf
			\EndFor
			\State	$\mathcal{R} = \mathcal{R}_{\text{min}}$
			\Until{$\overrightarrow{R_{k}} \in \mathbb{R}^{L}$ for every $k$}
			\State \Return $\mathcal{R}$
		\end{algorithmic}  
	\end{algorithm} 
	\paragraph{Complexity Analysis} 
	There are a total of $KL$ pick-up points to be appended into $\mathcal{R}$. For each addition, there are $\mathcal{O}(NK)$ candidates. The evaluation can be implemented by Sequential Evaluation in $\mathcal{O}((L+1)^{K})$. As a result, the overall complexity of Algorithm \ref{alg:4} is $\mathcal{O}(NK^{2}L(L+1)^{K})$.

\section{Experiments}

	\begin{figure*}
		\subfigure[A Comparison of the sum of PTT on real-world data ($N=25$)]{
			\includegraphics[width=0.32\linewidth]{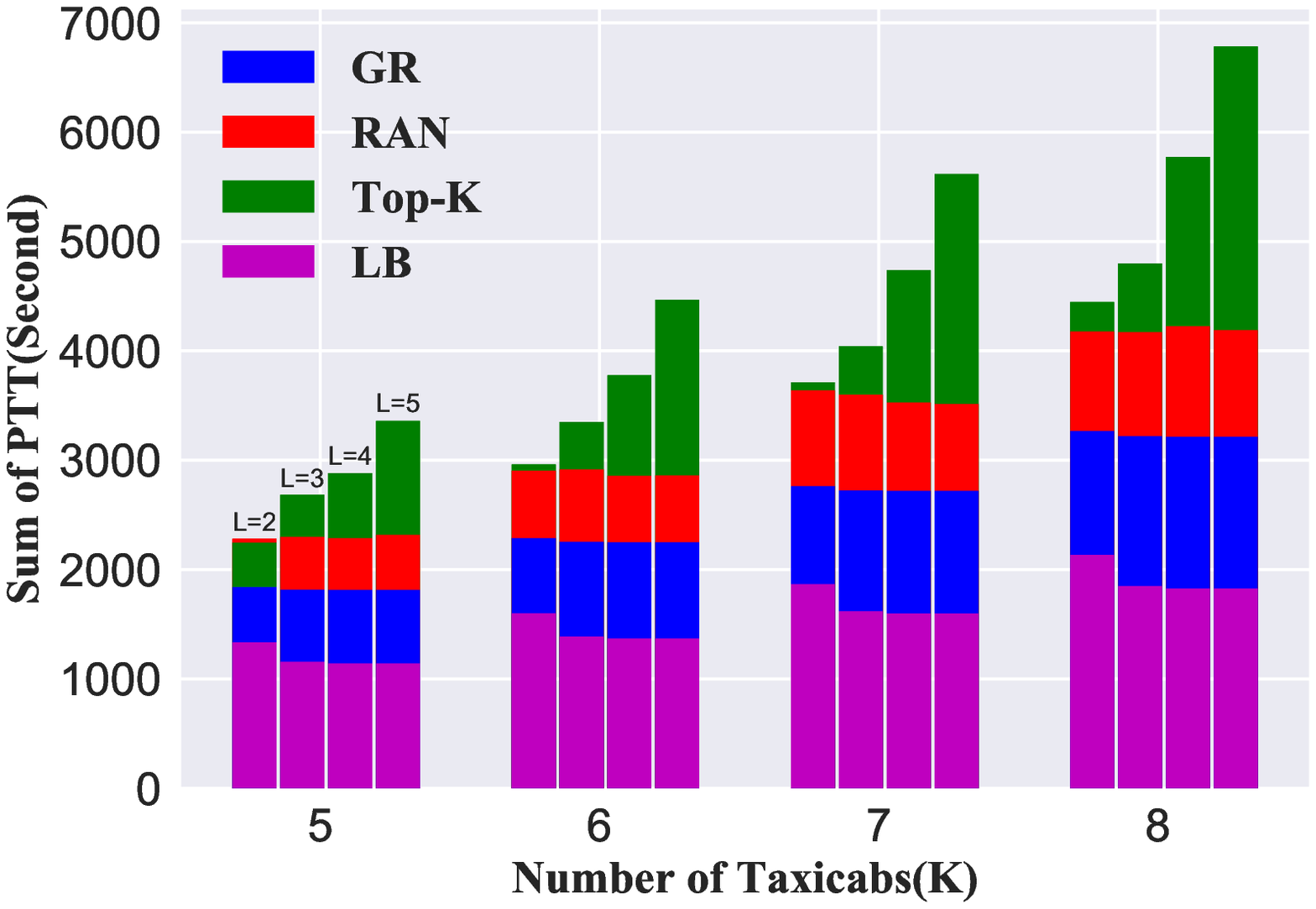}
			\label{fig:sim_epc}
		}
		\subfigure[A Comparison of the sum of cruising time]{
			\includegraphics[width=0.305\linewidth]{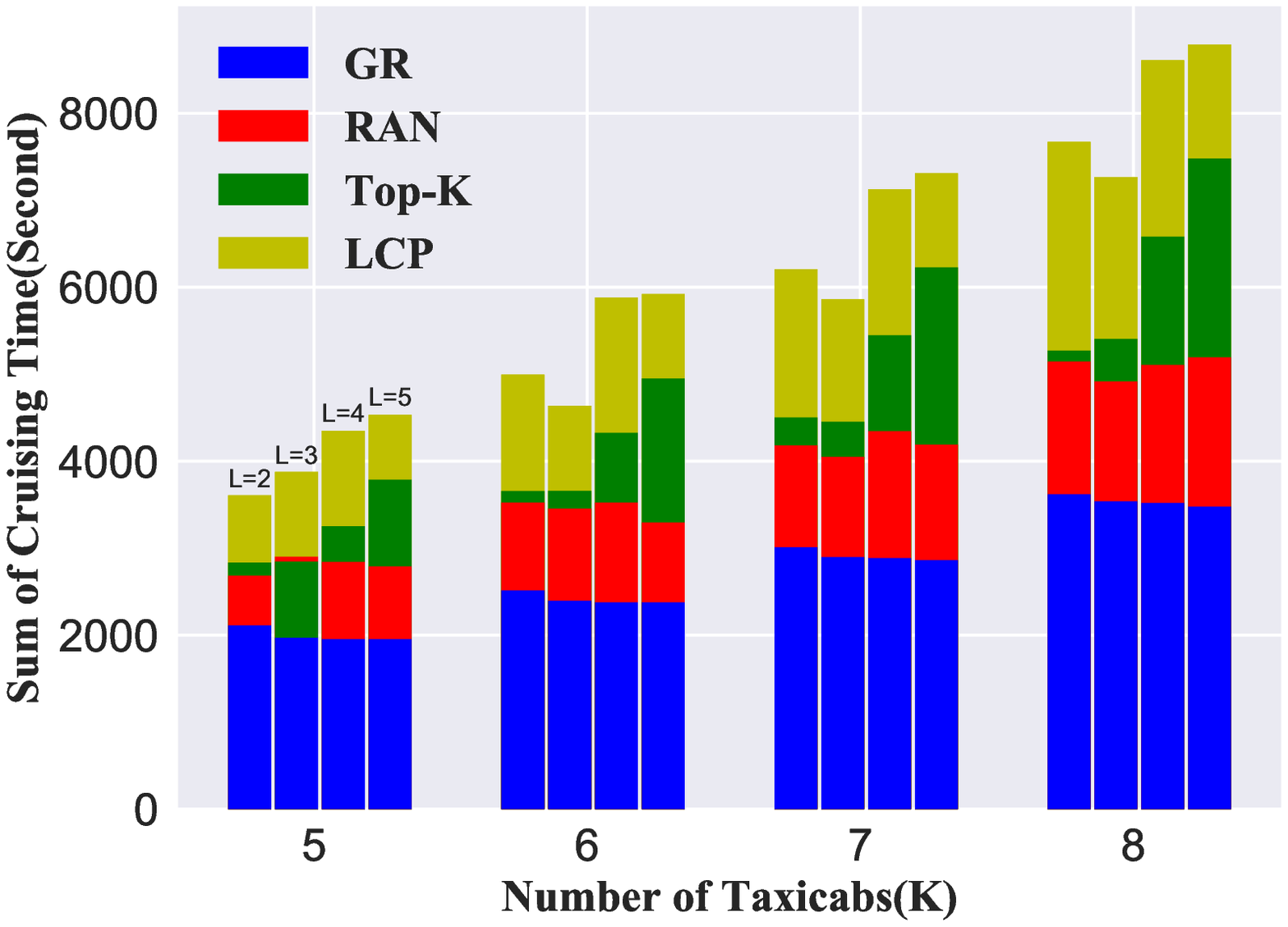}
			\label{fig:sim_time}
		}
		\subfigure[A Comparison of the number of passengers picked up]{
			\includegraphics[width=0.3\linewidth]{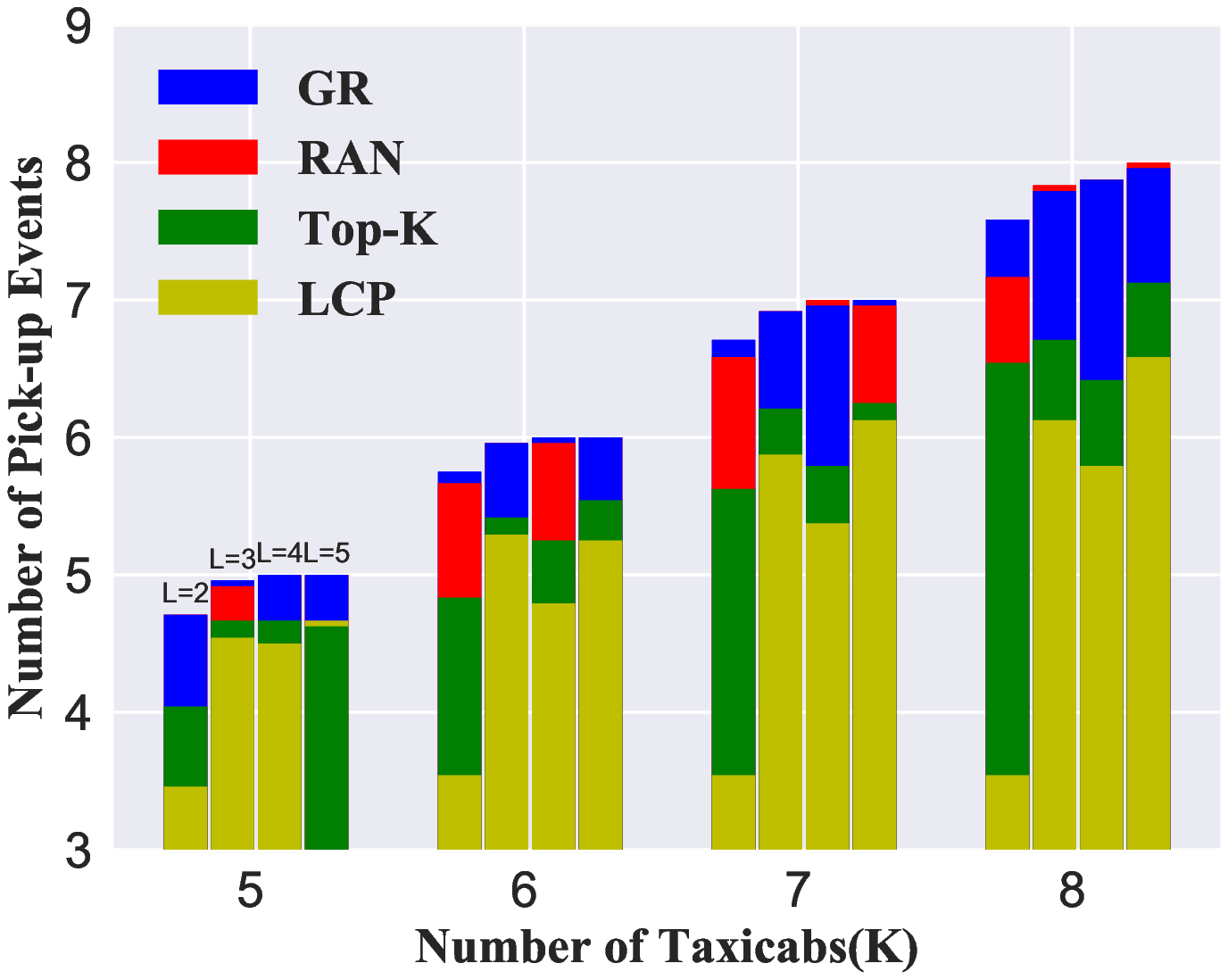}
			\label{fig:sim_num}
		}
		\caption{Experimental Results}
	\end{figure*}
%

Our experiments are conducted on a PC with an Intel Dual i7-4720 processor and 16 GB RAM. All the algorithms are implemented in C++11. The results of $\mathcal{F}$ and computation time by different approaches are compared. Based on real-world trace data, we demonstrate the effectiveness of our recommended routes in practice.

\paragraph{Real-world Data}
We adopt taxicab traces data set which is provided by Exploratorium - the museum of science, art and human perception \cite{crawdad}. The data set contains GPS records of 514 taxi drivers collected during 24 days in San Francisco Bay Area. Each record has four attributes: longitude, latitude, occupancy and timestamp. Based on historical data, we can extract a class of pick-up events. In the data processing, we extract records within the time period 6:00 PM-6:30 PM and find pick-up events within this time period to generate some clusters (by using the DBSCAN method).
	Subsequently, for each cluster $c$, we calculate the time intervals between two consecutive events and get an interval set $\{t_{j}\}$. We adopt an unbiased estimation for passenger arrival rate $\lambda_{c}$ as follows:
	\begin{equation}
	\small
	\lambda_{c} = \frac{n-1}{\sum t_{j}} ,\ (n=\left|\{t_{j}\}\right|)
	\end{equation}
	Afterwards, we represent pick-up points with centroids of clusters and use Google Map API to calculate the driving time between them. The start point $0$ for taxicabs is selected randomly on Google Map. In order to utilize trajectories comprehensively, we use cross-validation. Specifically, for each day, we use the other 23 days' data to estimate passenger arrival rates for pick-up points so that we can get 24 CMSR instances in total.
	
	
	\paragraph{Parameter Setting} 
	The maximum route length $L$ is set to 5, the same as \cite{Ge2010An}, which can restrict the maximum cruising time of driving route into a reasonable range. Due to an explosive increase of computation time, the maximum value of $K$ is set to 8. 
	$Penalty$ is set as an average value of traveling time matrix $T$, i.e. $764.3$ seconds.
	
	\paragraph{Compared Methods}
	\begin{table}[]
		\caption{Some Acronyms and Time Complexity}
  \addtolength{\tabcolsep}{-3.3pt}
		\centering
		\small
		\begin{tabular}{lll}
			\hline
			Acronym & Approach & Time Complexity\\ 
			\hline
			GR & Greedy Recommendation& $\mathcal{O}(NK^{2}L(L+1)^{K})$\\
			Top-K & Top-K Recommendation & $\mathcal{O}(LN^{L})$\\
			LB & Lower Bound for CMSR & $\mathcal{O}(LN^{L})$\\
			RAN & Random Selection of Routes& $\mathcal{O}(KL)$\\
			LCP & LCP Search \cite{Ge2010An} & $\mathcal{O}(LN^L)$\\
			\hline
			SA & Straightforward Approach& $\mathcal{O}(KL(L+1)^K)$\\
			SE & Sequential Evaluation& $\mathcal{O}((L+1)^K)$\\
			\hline
		\end{tabular}
		\label{tab:acr}
	\end{table}
	
In \S 3, we introduce several algorithms under different scenarios, which are summarized in Table \ref{tab:acr}.
	GR is our proposed method.
	Top-K recommendation is a scheme that recommends top $K$ routes with the minimum PTT to $K$ drivers but not considering the influence among routes.
	LB is calculated by multiplying $K$ and the minimum PTT.
	RAN is a random approach. It does not help reduce the cruising time for taxicabs but can balance the distribution of routes inherently.
	LCP method is introduced in \cite{Ge2010An} to solve MSR with a length constraint given pick-up probabilities. It uses a Round-Robin method to assign several best routes (5 routes in their paper) to taxicabs. 
	We compare this method with ours in simulation.
	SA and SE are subroutines for evaluating a given recommendation.
	
	
	\paragraph{A Comparison of the Sum of PTT}

	We conduct experiments on real-world data sets where $N$ is 25.
	Since LCP is not invented for CMSR, it is excluded from this part of experiments.
	 The results are shown in Figure \ref{fig:sim_epc}. As we can see, GR outperforms RAN and Top-K with a significant margin and is within a reasonable range from lower bound. On average, GR can find a recommendation with a smaller sum of PTT than RAN and Top-K by 22.4\% and 38.8\%, respectively. Its average solution gap with respect to LB is 61.8\%.
	 
	In addition, for a particular $K$, there is no obvious trend of the solutions generated by GR, RAN and LB as $L$ increases. We give a brief but intuitive explanation. When $L$ increases, the possible cruising time increases as well. At the same time, the probability of each driver picking up a passenger also increases, which leads to less penalty. Therefore, the sum of PTT is not necessarily monotonous with $L$. 
	Out of our expectation, Top-K behaves the worst. This is due to the correlation of those best routes. To be more specific, sorting routes in ascending order of PTT, the top $K$ routes may share a long prefix sub-route. In this case, one taxicab is very likely to arrive at a pick-up point following another taxicab immediately, which dramatically results in a great deal of cruising time. As a result, when $L$ increases, the sum of PTT of Top-K increases more significantly  than that of other approaches.
	
	\paragraph{Trace-driven Simulation}
	In order to evaluate the performance of the proposed approaches in practice, we test them via simulation based on real-world data. To do so, we replay the pick-up events for each day and test the recommended routes by different approaches. If a taxicab arrives at a pick-up point and satisfies the pick-up condition of CMSR, its cruising time is determined. If a taxicab fails to pick up a passenger along the whole route, the cruising time from the starting to the end location plus $Penalty$ is computed. The average results of 24 CMSR instances are recorded.\par
	We show the sum of cruising time and the number of picked up passengers for $K$ taxicabs in Figure \ref{fig:sim_time} and \ref{fig:sim_num}, respectively. GR performs much better than RAN, Top-K and LCP. In addition, the sum of cruising time of GR increases less intensely by $K$ than that of other methods. This is because GR can well balance the driving efficiency and the competition among drivers.
	LCP performs even worse than Top-K for the following reasons. On the one hand, Top-K can recommend more different routes than LCP. On the other hand, Top-K is built on CMSR model which characterizes the passenger pattern more accurately. In Figure \ref{fig:sim_num}, there is no red bar in some cases due to its overlap with blue bar, i.e., GR and RAN perform equally. 	
	\begin{table}[t]
		\caption{Running Time of GR with Two Evaluation Methods}
		\centering
		\small
		\begin{tabular}{ccccc}
			\hline
			$N$ & $K$ & $L$ & SE (second) & SA (second)\\
			\hline
			20 & 8 & 2 & 0.12 & 0.73\\
			20 & 8 & 3 & 1.16 & 11.02\\
			20 & 8 & 4 & 7.12 & 88.52\\
			20 & 5 & 5 &0.24 & 2.13\\
			20 & 6 & 5 &1.32 & 14.58\\
			20 & 7 & 5 &5.78 & 75.45\\
			10 & 8 & 5 &17.49 & 223.94\\
			15 & 8 & 5 &26.11 & 367.92\\
			20 & 8 & 5 & 32.58 & 492.46\\
			\hline
		\end{tabular}
		\label{tab:runtime}
	\end{table}

	\paragraph{A Comparison of Running Time for Two Evaluation Methods}
	Either SE or SA can be integrated into GR to evaluate a recommendation. To show the efficiency of SE compared with SA, we conduct several experiments on some synthetic data with different combinations of $N$, $K$ and $L$. The results are reported in Table \ref{tab:runtime}. It clearly shows that SE is about one order of magnitude faster than SA. 
	 
	\section{Conclusions and Future Work}
	There has been little effort dedicated to recommendation for a collection of taxicab drivers. In this paper, we propose CMSR, a collective mobile sequential recommendation, aimed to provide a set of routes to multiple taxicab drivers. The new metric of CMSR guarantees that the recommended routes can minimize the expected cruising time of taxicabs globally. In comparison with other methods, our method demonstrates its superior effectiveness and efficiency.

	In the future, we plan to explore more effective approximation by considering different passenger arrival/departure patterns and other advanced algorithms. Furthermore, we plan to extend CMSR to more practical problems which relax those hard restrictions, such as that a passenger missing one taxicab will not be picked up by next arriving taxicabs.
	
%
%

\bibliographystyle{named}

\begin{thebibliography}{1}
	\bibitem[\protect\citeauthoryear{Ge \bgroup \em et al.\egroup
	}{2010}]{Ge2010An}
	Yong Ge, Hui Xiong, Alexander Tuzhilin, Keli Xiao, Marco Gruteser, and
	Michael~J. Pazzani.
	\newblock An energy-efficient mobile recommender system.
	\newblock In {\em ACM SIGKDD International Conference on Knowledge Discovery
		and Data Mining}, pages 899--908, 2010.
	
	\bibitem[\protect\citeauthoryear{Huang \bgroup \em et al.\egroup
	}{2014}]{Huang2014Backward}
	Jianbin Huang, Xuejun Huangfu, Heli Sun, Hong Cheng, and Qinbao Song.
	\newblock Backward path growth for efficient mobile sequential recommendation.
	\newblock {\em IEEE Transactions on Knowledge and Data Engineering},
	27(1):46--60, 2014.
	
	\bibitem[\protect\citeauthoryear{Piorkowski \bgroup \em et al.\egroup
	}{2009}]{crawdad}
	Michal Piorkowski, Natasa Sarafjanovic-Djukic, and Matthias Grossglauser.
	\newblock Crawdad dataset epfl/mobility (v. 2009-02-24).
	\newblock \url{ http://crawdad.org/epfl/mobility/20090224}, 2009.
	
	\bibitem[\protect\citeauthoryear{Qian \bgroup \em et al.\egroup
	}{2015}]{Qian2015SCRAM}
	Shiyou Qian, Jian Cao, F~Mouël Le, Issam Sahel, and Minglu Li.
	\newblock Scram: A sharing considered route assignment mechanism for fair taxi
	route recommendations.
	\newblock In {\em ACM SIGKDD International Conference on Knowledge Discovery
		and Data Mining}, pages 955--964, 2015.
	
	\bibitem[\protect\citeauthoryear{Qu \bgroup \em et al.\egroup }{2014}]{Qu2014A}
	Meng Qu, Hengshu Zhu, Junming Liu, Guannan Liu, and Hui Xiong.
	\newblock A cost-effective recommender system for taxi drivers.
	\newblock In {\em ACM SIGKDD International Conference on Knowledge Discovery
		and Data Mining}, pages 45--54, 2014.
	
	\bibitem[\protect\citeauthoryear{Xun and Xue}{2011}]{Yun2011An}
	Yun Xun and Guangtao Xue.
	\newblock An online fastest-path recommender system.
	\newblock In {\em Knowledge Engineering and Management}, pages 341--348.
	Springer Berlin Heidelberg, 2011.
	
	\bibitem[\protect\citeauthoryear{Ye \bgroup \em et al.\egroup }{2018}]{Ye2018}
	Zeyang Ye, Lihao Zhang, Keli Xiao, Wenjun Zhou, Yong Ge, and Yuefan Deng.
	\newblock Multi-user mobile sequential recommendation: An efcient parallel
	computing paradigm.
	\newblock In {\em ACM SIGKDD International Conference on Knowledge Discovery
		and Data Mining}, pages 2624--2633, 2018.
	
	\bibitem[\protect\citeauthoryear{Yuan \bgroup \em et al.\egroup
	}{2013}]{Yuan2013T}
	Nicholas~Jing Yuan, Yu~Zheng, Liuhang Zhang, and Xing Xie.
	\newblock T-finder: A recommender system for finding passengers and vacant
	taxis.
	\newblock {\em IEEE Transactions on Knowledge and Data Engineering},
	25(10):2390--2403, 2013.
	
	\bibitem[\protect\citeauthoryear{Zhang and Tian}{2012}]{Zhang2012pCruise}
	Desheng Zhang and He~Tian.
	\newblock pcruise: Reducing cruising miles for taxicab networks.
	\newblock In {\em IEEE Real-time Systems Symposium}, 2012.
	
	\bibitem[\protect\citeauthoryear{Zheng \bgroup \em et al.\egroup
	}{2010}]{Zheng2010}
	Vincent~W. Zheng, Bin Cao, Yu~Zheng, Xing Xie, and Qiang Yang.
	\newblock Collaborative filtering meets mobile recommendation: A user-centered
	approach.
	\newblock In {\em Proceedings of the Twenty-Fourth AAAI Conference on
		Artificial Intelligence}, 2010.
	
	\bibitem[\protect\citeauthoryear{Zheng \bgroup \em et al.\egroup
	}{2011}]{where-to-find-my-next-passenger}
	Yu~Zheng, Xing Xie, Guangzhong Sun, Liuhang Zhang, Jing Yuan, and Nicholas~Jing
	Yuan.
	\newblock Where to find my next passenger?
	\newblock In {\em Proceedings of the 13th ACM International Conference on
		Ubiquitous Computing (Ubicomp 2011)}, September 2011.
	
\end{thebibliography}

\end{document}